\documentclass[runningheads]{llncs}

\usepackage[T1]{fontenc}
\usepackage{amsfonts}
\usepackage{braket}
\usepackage{graphicx}
\usepackage{amsmath}
\usepackage{enumitem}
\usepackage[colorlinks = true, citecolor=blue, linkcolor=blue, urlcolor=blue]{hyperref}
\usepackage{multirow}

\usepackage[table]{xcolor}
\definecolor{headerblue}{RGB}{211, 211, 244} % Custom color definition

\newcommand{\tr}{\textup{Tr}}

\usepackage{float}
\begin{document}

\title{Collusion-Secure Semi-Quantum Secret Sharing Scheme using a Quantum Third Party}
\titlerunning{Collusion-Secure SQSS}

\author{
  Santanu Majhi\thanks{ \texttt{\href{Santanu}{santanum\_r@isical.ac.in}}} 
  \and 
  Nirupam Basak\thanks{\texttt{\href{Nirupam}{nirupambasak2020@iitkalumni.org}}}
}

\institute{Cryptology and Security Research Unit, Indian Statistical Institute, Kolkata}

\maketitle

% ABSTRACT %
\begin{abstract}

     The security of a secret sharing protocol is compromised if one or more dishonest participants can reconstruct the secret by deviating from the prescribed protocol, particularly when such malicious behavior remains undetected. Semi-Quantum Secret Sharing (SQSS) is a variant of secret sharing in which the participants possess only classical capabilities, such as preparing and measuring qubits in the computational ($Z$) basis, while a quantum-capable third party assists the dealer in generating and distributing quantum shares of a classical secret. The limited quantum capabilities of the participants in SQSS protocols may introduce vulnerabilities to collusion attacks, enabling an assisting quantum third party, in collaboration with a dishonest classical participant, to jointly recover the secret even without detection. In this work, we propose a novel SQSS protocol that eliminates this vulnerability and achieves information-theoretic security against collusion attacks. We further prove that the proposed protocol is secure against a broad class of external and internal attacks. Finally, a comparative analysis demonstrates that our construction advances existing SQSS protocols by simultaneously optimizing qubit efficiency, involves a classical dealer, and resilience against DCNA attacks using basic Bell states.

    \keywords{Semi-quantum protocol, Secret Sharing Scheme, Collusion attack, Security, Improved protocol, DCNA secure.}
\end{abstract}

\section{Introduction} 
Quantum key distribution (QKD) is a fundamental primitive in quantum cryptography that enables the secure distribution of cryptographic keys based on the principles of quantum mechanics. The first QKD protocol, BB84~\cite{bennett1984quantum,PhysRevLett.85.441}, established a secure key-sharing mechanism between two parties, Alice and Bob, who possess quantum capabilities. A natural extension of this framework is to consider scenarios involving multiple participants, where full access to a shared key should not depend on the trustworthiness of a single participant. This motivation led to the development of Quantum Secret Sharing (QSS)~\cite{PhysRevLett.83.648}, which extends the idea of secure key distribution to multiple participants such that only an authorized subset can jointly reconstruct the secret. The seminal QSS scheme proposed by Hillery \emph{et al.}~\cite{PhysRevA.59.1829} provided an early realization of this concept. Subsequently, numerous QSS schemes~\cite{CHATTOPADHYAY2024100608,PhysRevA.72.022303,PhysRevA.71.022321,Xie2015,10.1007/s11128-012-0379-6,choudhury2009quantum,nie2011quantum} have been investigated, which can broadly be classified into two categories: schemes distributing a classical secret via quantum shares, and schemes distributing a quantum secret via quantum shares.

However, for both cases conventional QSS protocols generally assume that the participants possess quantum capabilities and can perform quantum operations on their shares. This raises a natural question: can secret sharing be realized in an intermediate operational model where the dealer and the other participants are restricted to perform only simple quantum operations to run the protocol? This question motivates the semi-quantum paradigm. In 2007, Boyer \emph{et al.}~\cite{PhysRevLett.99.140501} introduced the concept of Semi-Quantum Key Distribution (SQKD), in which Alice possesses full quantum capabilities, whereas Bob is restricted to a limited set of operations, such as measuring qubits in the computational ($Z$) basis and  prepare fresh qubit in the computational basis and transmit it. Building upon this paradigm the first Semi-Quantum Secret Sharing (SQSS) protocol~\cite{PhysRevA.82.022303} have subsequently been developed, where the classical participants possess restricted quantum capabilities: they can measure qubits in the computational basis $\{\ket{0}, \ket{1}\}$, reorder qubits using delay lines, prepare fresh computational-basis states, or transmit qubits undisturbed.

\subsection{Related Works}
A variety of SQSS protocols have been established to lower the technical barrier for classical participants. Early efforts explored product-state protocols by Li \emph{et al.}~\cite{Li2018} and simplified particle-rearrangement schemes by Xin \emph{et al.}~\cite{doi:10.1142/S0217732324500846}, followed by Hou \emph{et al.}'s hybrid architecture~\cite{Hou2024Circular} employing single and GHZ states under a cyclic transmission framework. To address practical deployment challenges, Li \emph{et al.}~\cite{Li2024Collective} introduced decoherence-free states to mitigate quantum noise, and later Tian \emph{et al.} developed Bell-state encoding strategies~\cite{Tian2021}, which were subsequently fortified against specific security vulnerabilities~\cite{He2024}. Together with broader efforts across the field, these developments continue to streamline quantum operations for resource-constrained classical entities.

The work ``Security Analysis and Improvement of a Semi-Quantum Secret Sharing Protocol with Specific Bits Based on a Third Party'' by He \emph{et al.}~\cite{He2026} presents a SQSS protocol for sharing a classical secret among two parties, namely Bob and Charlie, with Alice as the dealer.  The protocol is proposed as an improvement over the scheme presented by Ma \emph{et al.}~\cite{Ma_2024}, which is vulnerable to Double CNOT attack (DCNA). In the protocol of He \emph{et al.}~\cite{He2026}, the participants are assumed to have limited quantum capabilities, whereas two fully quantum-capable third parties, referred to as $\mathrm{Trent}_1$ and $\mathrm{Trent}_2$ (jointly Trent), assist the dealer Alice in preparing the quantum shares. The protocol aims to achieve secure secret sharing while reducing the quantum resources required of the communicating parties.

In the aforementioned protocol~\cite{He2026}, the possibility that a fully quantum-capable third party may collude with a classical participant during the secret reconstruction phase is not considered. Specifically, the protocol is vulnerable to a scenario where the two third parties, Trent, represent a single adversary colluding directly with the classical participant, Bob. Under this assumption, Trent can launch an intercept-and-resend attack on the quantum channels between Alice and participant Charlie. Specifically, Trent intercepts the transmitted quantum shares, measures them on the computational basis $(Z)$, and subsequently forwards the corresponding measurement results to the intended recipient. Since the protocol lacks sufficient mechanisms to detect such malicious behavior, Trent remains undetected throughout the execution of the protocol. Consequently, this vulnerability may compromise the confidentiality of the shared secret and pose a potential threat to the overall security of the protocol.

\subsection{Our Contribution}
% We identify and analyze a critical security vulnerability in the SQSS protocol proposed by He \emph{et al.}~\cite{He2026}, demonstrating its susceptibility to eavesdropping and manipulation strategies that compromise the underlying security guarantees. To resolve these flaws, we propose an enhanced version of the SQSS protocol that strictly preserves the operational definition of semi-quantumness, ensuring classical participants require only basic $Z$-basis measurement and reflection capabilities while guaranteeing robust key distribution. We formally prove that the proposed scheme prevents a wide class of active and passive threats, including DCNA, Decoy State Substitution, and participant collusion attacks by leveraging diagonal decoy states and permuted basis positions with delayed operational mode announcements to guarantee information-theoretic security.
We identify a critical security vulnerability in the SQSS protocol of He \emph{et al.}~\cite{He2026} that exposes it to eavesdropping and manipulation. To resolve these flaws, we propose an enhanced SQSS protocol that restores robust security keeping the definition of SQSS protocol intact.
The main contributions of this paper are summarized as follows:
\begin{itemize}
    \item \textbf{Vulnerability Analysis of Existing Protocol:} We uncover a critical security flaw in the SQSS protocol by He~et~al.~\cite{He2026}, demonstrating how undetectable collusion between a dishonest classical participant (Bob) and untrusted quantum third parties (Trent) allows them to extract the secret information of another participant (Charlie) without detection, ultimately enabling them to successfully generate the shared secret bit without any contribution from Charlie.
    
    \item \textbf{Enhanced SQSS Protocol Design:} To address these vulnerabilities, we propose an improved SQSS protocol that strictly adheres to the operational definition of semi-quantumness, requiring classical participants to possess only $Z$-basis measurement and reflection capabilities while facilitating secure key distribution.
    
    \item \textbf{Provable Security Guarantees:} We formally prove that the proposed scheme provides robust protection against a broad spectrum of active and passive attacks, including Double CNOT attack (DCNA), decoy state substitution, and participant collusion. Information-theoretic security against any collusion whether involving classical participant or quantum entities, is strictly guaranteed through the integration of diagonal decoy states, permuted basis positions, and delayed operational mode announcements.
\end{itemize}

\section{Analysis of He \emph{et al.}'s Protocol}
The objective of the SQSS protocol is for Alice (a semi-quantum dealer) to share a classical secret binary sequence $S = \{s_1, s_2, \dots, s_n\}$ with Bob and Charlie (classical parties) using the assistance of two quantum servers, $\mathrm{Trent}_1$ and $\mathrm{Trent}_2$.
A brief overview of the main steps involved in the protocol proposed by He \emph{et al.}~\cite{He2026} is provided below. The reader is referred to the original work for a detailed description.

\paragraph{\textbf{Step I-1 \emph{(Setup and Classical Encryption):}}}
Alice randomly prepares two binary sequences: $K_B = \{k_1^B, k_2^B,            \dots, k_n^B\} \text{ and }
         K_C = \{k_1^C, k_2^C, \dots, k_n^C\}$, where $k_i^B, k_i^C \in \{0, 1\}$. 
 She computes the classical encryption key $K := K_B \oplus K_C$.
 
Alice encrypts the secret sequence $S$ into a ciphertext sequence $S_T:=S\oplus K$
% \[
%     S_T = S \oplus K = \{s_1 \oplus k_1^B \oplus k_1^C, \dots, s_n \oplus k_n^B \oplus k_n^C\}
% \]
and sends it to $\mathrm{Trent}_1$ through a secure classical channel.

\paragraph*{\textbf{Step I-2 \emph{(Quantum Encoding and $\mathrm{\mathbf{Trent}}_1$ Decoy Insertion):}}}
Upon receiving $S_T$, $\mathrm{Trent}_1$ encodes each bit $s_i^T \in S_T$ into an entangled EPR pair sequence $T = \{t_1, t_2, \dots, t_n\}$, according to the rule
\begin{equation*}
\ket{t_i}_{BC} = 
\begin{cases} 
\frac{1}{\sqrt{2}}(\ket{00} + \ket{11})_{BC}, & \text{if } s_i^T = 0 \\ 
\frac{1}{\sqrt{2}}(\ket{01} + \ket{10})_{BC}, & \text{if } s_i^T = 1 
\end{cases}
\end{equation*}
 $\mathrm{Trent}_1$ splits $T$ into two single-particle sequences, $T_B=\{t^B_1,t^B_2,\ldots,t^B_n\}$ and $T_C=\{t^C_1,t^C_2,\ldots,t^C_n\}$. He inserts $m$ randomly selected decoy particles from $\{\ket{0},\ket{1},\ket{+},\ket{-}\}$ into $T_B$ and $T_C$, forming $T'_B$ and $T'_C$, respectively, and sends them to $\mathrm{Trent}_2$.

\paragraph*{\textbf{Step I-3 \emph{($\mathrm{\mathbf{Trent}}_2$ Decoy Insertion):} }}
Fter recieving $T'_B$ and $T'_C$, $\mathrm{Trent}_2$ prepares his own set of $m$ decoy particles from $\{\ket{0}, \ket{1}, \ket{+}, \ket{-}\}$ and randomly inserts them into the sequences, yielding $T''_B$ and $T''_C$, respectively, before sending them to Alice. Only $\mathrm{Trent}_2$ knows the positions and states of his own decoys. 

\paragraph*{\textbf{Step I-4 \emph{(Alice's Key Injection and Distribution):} }}
$\mathrm{Trent}_1$ and $\mathrm{Trent}_2$ publish the exact index positions of their respective decoy particles. Alice prepares quantum state sequences $\ket{K_B} = \{\ket{k_i^B}\}_{i}$ and $\ket{K_C} = \{\ket{k_i^C}\}_{i}$ representing her classical keys $K_B$ and $K_C$. %using the $Z$-basis:
% \begin{equation*}
% \ket{k_i^{B}}, \ket{k_i^{C}}  = 
% \begin{cases} 
% \ket{0}, & \text{if } k_i^{B}, k_i^{C} = 0 \\ 
% \ket{1}, & \text{if } k_i^{B}, k_i^{C} = 1 
% \end{cases}
% \end{equation*}
% Alice leaves the decoy perticles untouched. 
She randomly inserts the particles of $\ket{K_B}$ and $\ket{K_C}$ into the random positions of $T''_B$ and $T''_C$, respectively, forming the final sequences $T'''_B$ and $T'''_C$ and sends $T'''_B$ to Bob and $T'''_C$ to Charlie.

\paragraph*{\textbf{Step I-5 \emph{(Eavesdropping Detection and Verification):} }}
Alice declares the original positions of the inserted decoy particles to Bob and Charlie. For each decoy particle, Bob and Charlie randomly perform a \texttt{MEASURE} or \texttt{REFLECT} operation as follows.
\begin{itemize}
    \item \textbf{Case 1 (\texttt{MEASURE}):} Bob (Charlie) measures the particle in the $Z$-basis. The corresponding Trent verifies the outcome according to the preparation basis. \item \textbf{Case 2 (\texttt{REFLECT}):} Bob (Charlie) returns the particle undisturbed through an optical delay line. The corresponding Trent measures it in its preparation basis to verify its integrity.
\end{itemize}
If the calculated error rate exceeds a predefined threshold, the protocol aborts.

\paragraph*{\textbf{Step I-6 \emph{(Measuring and Secret Reconstruction):} }}
Bob and Charlie measure all remaining non-decoy particles in the $Z$-basis, obtaining raw measurement lists $MR_B$ and $MR_C$, respectively. Alice then publicly reveals the exact coordinates of the keys $\ket{K_B}$ and $\ket{K_C}$. 
Bob (Charlie) extracts $K_B\ (K_C)$ from $MR_B\ (MR_C)$. They collaborate to reconstruct $K = K_B \oplus K_C$ and $S_T$ from the remaining measurement results. Using $K$, they decrypt $S_T$ to obtain the original secret sequence $S = S_T \oplus K$.%, while checking correlations on the remaining EPR particles ($MR'_B \oplus MR'_C$) to guarantee correctness.

% \section{The Inside-Collusion Attack}
% A fatal flaw exists when a single adversarial entity controls the classical participant Bob as well as both quantum servers, $\mathrm{Trent}$ and $\mathrm{Trent}_2$ ($\text{Bob} = \text{$\mathrm{Trent}$} = \text{$\mathrm{Trent}_2$}$). This allows Bob to bypass security entirely, eavesdrop on Charlie's line, and steal the secret without detection.

\section{The Inside-Collusion Attack \label{sec:Attack_scenario}}
A fatal flaw exists when a single adversarial entity controls one of the classical participant, say Bob, as well as both quantum servers, $\mathrm{Trent}_1$ and $\mathrm{Trent}_2$ (that is, $\text{Bob} = \text{$\mathrm{Trent}_1$} = \mathrm{Trent}_2$)%\footnote{While the vulnerability can also be exploited through collusion between any Trent and Bob, we present a more generalized scenario to clearly illustrate its full impact.}
. This allows Bob to bypass security entirely by executing a universal $Z$-basis measurement sweep on Charlie's line, using his status as the verifying authority to falsify error checks.
% =====================
% \footnotetext[1]{}
% =====================

\subsection{The Adversary's Strategy and Universal Z-Sweep}
Since Alice randomly inserts her key sequences $K_B$ and $K_C$ among the decoy particles, leaving the original decoy positions unknown to Bob during transmission, Bob disregards tracking individual particle positions and instead adopts the following strategy during Step I-4.

\begin{enumerate}
    \item \textbf{Universal Interception and Measurement:} When Alice sends $T'''_C$ to Charlie, Bob intercepts the entire quantum stream 
    % . Without attempting to distinguish between decoy particles, information particles, or Alice's keys, Bob
    and measures \emph{every single qubit} in the $Z$-basis.
    \item \textbf{Impact on the Quantum States:} 
    \begin{itemize}
        \item For Alice's key bits ($\ket{K_C}$) and the information carriers ($T_C$), which would be measured in the $Z$-basis by Charlie in Step I-6, Bob's measurement gives the exact bit values with $100\%$ accuracy without causing detectable state collapse.
        \item For the diagonal decoy particles ($\ket{+}, \ket{-}$) originally generated by the Trents, Bob's $Z$-basis measurement forces an irreversible collapse into either $\ket{0}$ or $\ket{1}$ with a $50\%$ probability. 
       \item The $Z$-basis measurement of Bob does not disturb the decoy particles $\{\ket{0}, \ket{1}\}$.
    \end{itemize}
    \item \textbf{Forwarding to Charlie:} Bob logs the measurement results for the entire sequence in order and forwards the post-measurement states (which now consist purely of $\ket{0}$ and $\ket{1}$ states) to Charlie. Charlie receives them with no knowledge of the interception.
\end{enumerate}

\subsection*{Falsification of the Security Check and Secret Theft}
During the eavesdropping detection phase (Step I-5), as the verifying authority (Trent) is colluding with the attacker (Bob), detection is completely neutralized as follows.
\begin{itemize}
    \item \textbf{Bypassing Case 1 (MEASURE):} When Charlie performs a $Z$-basis measurement on a decoy particle, he measures the state already measured by Bob in the same basis. When Charlie sends his measurement result to Trent for verification, Trent (who is Bob) automatically agrees with Charlie's result, reporting an error rate of $0\%$.
    \item \textbf{Bypassing Case 2 (REFLECT):} When Charlie reflects a decoy particle back untouched, it arrives at Trent in a collapsed $Z$-basis state. Instead of performing a legitimate check that would reveal the collapse, Trent simply ignores the state discrepancy and falsely reports that the state was preserved perfectly. 

    Thus, being the sole verifying authority enables Trent to trivially pass the quantum bit error rate (QBER) check.
    
    % Furthermore, because Trent generates the decoy basis used to construct $T''_C$, he possesses complete knowledge of the decoy sequence, enabling him to trivially pass the quantum bit error rate (QBER) check.
\end{itemize}

\subsection*{Key Reconstruction Dishonestly}
In Step I-6, Alice publicly reveals the exact positions at which she randomly inserted $K_C$, allowing Charlie to extract the $K_C$ bits from his measurement results, while Bob uses Alice’s announced positions to identify the corresponding indices in his universal $Z$-basis measurement record. Bob then recovers the $K_C$ sequence and combines it with the directly received $K_B$ sequence to reconstruct the complete encryption key $K = K_B \oplus K_C$.
   
Bob decrypts the classical ciphertext $S_T$, which is already known to $\mathrm{Trent}_1$ and also can be recovered from his remaining measurement records, to obtain Alice's original secret $S = S_T \oplus K$. Consequently, Charlie is successfully excluded from the secret-sharing process, while Alice remains entirely unaware of the security breach.

A detailed example demonstrating the correctness of this attack is provided in Appendix~\ref{appendix:A}.

\section{The Improved SQSS Protocol}

Our proposed scheme directly addresses the vulnerability identified in the previous protocol by minimizing the authority granted to trusted third parties, thereby mitigating collusion risks with rogue participants. Specifically, the protocol is structured to allow the dealer, Alice, to actively intervene at key stages, preventing Trent from deceiving either the classical parties or Alice herself. Unlike the previous protocol, this scheme incorporates only one trusted third party, Trent. For security reasons, the authority to insert decoy states, retain knowledge of their positions, and perform the QBER verification is restricted exclusively to Alice, effectively eliminating key vulnerabilities. 

Prior to protocol execution, Alice initializes the system by securely sharing secret permutations $\Pi_B$ and $\Pi_C$ with Bob and Charlie, respectively. This classical pre-sharing step maintains full compliance with the semi-quantum framework, as classical participants are restricted only in their quantum capabilities, not in classical index processing. The detailed protocol is as follows:
\subsection*{1. Protocol Architecture and Initial Setup}

\begin{itemize}[leftmargin=*]
    \item \textbf{Participants:}
    \begin{itemize}
        \item \textbf{Alice ($A$):} The dealer who holds the classical secret $S$.
        \item \textbf{Bob ($B$) and Charlie ($C$):} Classical (semi-quantum) participants capable only of measuring or preparing in the computational $Z$-basis ($\{|0\rangle, |1\rangle\}$) or directly reflecting particles.
        \item \textbf{$\mathrm{\mathbf{Trent}}$ ($T$) :} Fully quantum third party equipped with arbitrary quantum state preparation and entanglement measurement capabilities.
    \end{itemize}
    \item \textbf{Secret Sequence:} $S = \{s_1, s_2, \dots, s_n\}$, where $s_i \in \{0, 1\}$.
\end{itemize}

\subsection*{2. Phase-Wise Protocol Execution}

\paragraph{\textbf{\emph{Phase 1:}} Classical Key Generation and Ciphertext Construction}
\begin{enumerate}[label=\textbf{Step 1.\arabic*}, leftmargin=*]
    \item \textbf{Local Key Generation:} Alice randomly generates two independent binary sequences of length $n$
    \[
    K_B = \{k_1^B, \dots, k_n^B\}, K_C = \{k_1^C, \dots, k_n^C\} \text{where } k_i^B, k_i^C \xleftarrow\$ \{0, 1\}.
    \]
    \item \textbf{Master Key XOR:} Alice computes the combined master key sequence
    \[
    K = K_B \oplus K_C = \{k_1, k_2, \dots, k_n\},\text{where } k_i = k_i^B \oplus k_i^C.
    \]
    \item \textbf{Secret Encryption:} Alice masks her secret $S$ using $K$ to create the ciphertext string
    \[
    S_T = S \oplus K = \{s^T_1, s^T_2, \dots, s^T_n\}, \text{where } s^T_i = s_i \oplus k_i,
    \]
    and then securely transmits $S_T$ to quantum third-party $\mathrm{Trent}$.
\end{enumerate}

\paragraph{\textbf{\emph{Phase 2:}} Quantum Entanglement Encoding by $\mathrm{Trent}$}
\begin{enumerate}[label=\textbf{Step 2.\arabic*}, leftmargin=*]
    \item \textbf{Bell State Encoding:} Upon receiving $S_T$, $\mathrm{Trent}$ constructs $n$ pairs of entangled sequence $T = \{|t_1\rangle, |t_2\rangle, \dots, |t_n\rangle\}$ based on $S_T$ as follows
    \[
    |t_i\rangle := \begin{cases} 
    |\Phi^+\rangle = \frac{1}{\sqrt{2}}(|00\rangle + |11\rangle)_{BC}, & \text{if } s^T_i = 0 \\ 
    |\Psi^+\rangle = \frac{1}{\sqrt{2}}(|01\rangle + |10\rangle)_{BC}, & \text{if } s^T_i = 1 .
    \end{cases}
    \]
    \item \textbf{Sequence Splitting:} $\mathrm{Trent}$ splits the entangled pairs into two single-particle streams.
    \begin{itemize}
        \item $T_0$: Contains particle $1$ of every pair $|t_i\rangle$.
        \item $T_1$: Contains particle $2$ of every pair $|t_i\rangle$.
    \end{itemize} 
    \item $\mathrm{Trent}$ transmits both quantum sequences $T_0$ and $T_1$ directly to Alice.
    % Then Alice decides who will get $T_0$ and who gets $T_1$.

    \item Trent prepares four qubit sequences, each consisting entirely of one of the states $\{\ket{0}, \ket{1}, \ket{+}, \ket{-} \}$
    % \[
    % |+\rangle = \frac{1}{\sqrt{2}}(|0\rangle + |1\rangle), \quad |-\rangle = \frac{1}{\sqrt{2}}(|0\rangle - |1\rangle),\]
     and sends them to Alice.
\end{enumerate}
% Note:  It is not fixed who will receive $T_0$ or $T_1$.

\paragraph{\textbf{\emph{Phase 3:}} Quantum Key Injection and Decoy Insertion by Alice}
\begin{enumerate}[label=\textbf{Step 3.\arabic*}, leftmargin=*]
    \item \textbf{Quantum Key Preparation:} Based on $K_B$ and $K_C$, Alice prepares single-qubit quantum states in the $Z$-basis:
    \[
    |K_B\rangle = \{|k_1^B\rangle, |k_2^B\rangle, \dots, |k_n^B\rangle\}, \quad |K_C\rangle = \{|k_1^C\rangle, |k_2^C\rangle, \dots, |k_n^C\rangle\}
    \]
    \item \textbf{Key Injection:} Alice decides who will get $T_0$ and who gets $T_1$. The sequence of particles Bob gets renamed as $T_B$ and the same for Charlie is $T_C$. Alice interposes $|K_B\rangle$ into $T_B$ and $|K_C\rangle$\ into $T_C$ at randomized indices, creating expanded sequences $T'_B$ and $T'_C$.
    \item \textbf{{Decoy State Insertion:}} Alice inserts $m$ decoy particles chosen randomly from the BB84 basis set $\{|0\rangle, |1\rangle, |+\rangle, |-\rangle\}$
    into the $T'_B$ and $T'_C$ to obtain $T''_B$ and $T''_C$, a longer sequence.  
    % \textcolor{red}{Here Trent informs Alice about the decoy positions.}

    \item Alice sends sequence $T''_B$ to Bob and $T''_C$ to Charlie across quantum channels. 
\end{enumerate}

\paragraph{\textbf{\emph{Phase 4:}} Semi-Quantum Operations and Sequence Permutation}

\begin{enumerate}[label=\textbf{Step 4.\arabic*}, leftmargin=*]
    \item \textbf{Decoy Notification:} Once receipt of $T''_B$ and $T''_C$ is confirmed, Alice tells the index positions designated as decoy slots to Bob and Charlie via classical channel.
     \item \textbf{Sequence Rearrangement:} Bob and Charlie apply their pre-shared secret permutations $\Pi_B$ and $\Pi_C$ (shared secretly by Alice prior to protocol execution) strictly on the decoy sequence.
    % to prevent intercept-resend attacks.
    
    \item \textbf{Operation Choice:} For every decoy position, Bob and Charlie randomly and independently select one of two semi-quantum modes:
    \begin{itemize}
        \item {$\mathrm{\mathbf{SIFT}}$ Mode:} Measure the qubit in the $Z$-basis and forward the resulting state to $\mathrm{Trent}$.
        \item {$\mathrm{\mathbf{CTRL}}$ Mode:} Directly reflect the particle back to $\mathrm{Trent}$ without disturbance.
    \end{itemize}
    Bob and Charlie must use a delay line to maintain the ordering. Since $\Pi_B$ and $\Pi_C$ are unknown to Trent, the physical order of returned decoys is completely scrambled to any eavesdropper.

    % \item Bob and Charlie send the sequence to Trent.
\end{enumerate}

\paragraph{\textbf{\emph{Phase 5:}} Decoy Checking and Eavesdropping Verification}
\begin{enumerate}[label=\textbf{Step 5.\arabic*}, leftmargin=*]

    \item \textbf{Mode and Basis Announcement:} After Trent confirms the arrival of the decoy sequences, Alice announces the permuted basis sequences (applying $\Pi_B$ for the basis sequence of Bob, similar for Charlie) without revealing the actual state.
    % Then Trent measures the respective sequences received from Bob and Charlie.
    
    % \item  Alice publish the basis he selects to fill the decoy positions.
    % operation mode choices ($\text{$\mathrm{SIFT}$}$ vs. $\text{$\mathrm{CTRL}$}$) alongside their measurement results for all $Z$-basis decoy checks.

    \item Trent measures accordingly and declares the measurement results.

    \item Bob and Charlie announce their operation mode choices ($\text{$\mathrm{SIFT}$}$ vs. $\text{$\mathrm{CTRL}$}$) alongside their measurement results for all $Z$-basis decoy checks.
     % \item \textbf{Sequence Restoration:} Bob and Charlie un-shuffle their respective sequences using their inverse permutations $\Pi_B^{-1}$ and $\Pi_C^{-1}$.
    
    \item \textbf{QBER Computation:} Alice evaluates the Quantum Bit Error Rate (QBER) by comparing $\mathrm{Trent}$, Bob and Charlie's reported outputs against her initial decoy preparations ($|0\rangle, |1\rangle, |+\rangle, |-\rangle$).  
    % \textcolor{red}{Alice receives both the decoys from Trent and Bob's and Charlie's outputs from them. Then Alice checks if it matches or not in sufficient places to reach up to the QBER.}
    \item \textbf{Abort/Proceed Threshold:}
    \begin{itemize}
        \item If $\text{QBER} > \text{QBER}_{\text{threshold}}$, eavesdropping or a malicious third-party probe attack is detected; the protocol \textbf{aborts}.
        \item If $\text{QBER} \le \text{QBER}_{\text{threshold}}$, Alice, Bob, and Charlie proceed to secret reconstruction.
    \end{itemize}
\end{enumerate}

\paragraph{\textbf{\emph{Phase 6:}} Key Extraction and Secret Reconstruction}
\begin{enumerate}[label=\textbf{Step 6.\arabic*}, leftmargin=*]
    \item \textbf{Key Measurement and Extraction:} Bob and Charlie perform $Z$-basis measurements on the remaining particles, obtaining raw measurement lists $MR_B$ and $MR_C$, respectively.
    \item \textbf{Key Index Announcement:} Once the QBER is verified to be within the acceptable threshold, Alice publicly announces the specific positions corresponding to the key particles $|K_B\rangle$ and $|K_C\rangle$ to Bob and Charlie, respectively.% over the classical channel.
    \item \textbf{Classical Key Combination:} Bob and Charlie extract their respective classical key sequences $K_B$ and $K_C$ and communicate over an authenticated classical channel to perform bitwise XOR operations on their extracted keys    $K = K_B \oplus K_C $. The remaining measurement results are denoted as $MR'_B$ and $MR'_C$ for Bob and Charlie, respectively.
    \item \textbf{Final Secret Decryption:} Bob and Charlie collaborate to calculate $S_T=MR'_B\oplus MR'_C$. Using this $S_T$, Bob and Charlie reconstruct the original secret $S$
    \[
    S = S_T \oplus K.% = (S \oplus K) \oplus K .
    \]
\end{enumerate}
% ====================

% \begin{tcolorbox}[protocolbox, title={Protocol Architecture: Semi-Quantum Secret Sharing}]
%

\begin{figure}[ht]
\centering
\includegraphics[width=\textwidth]{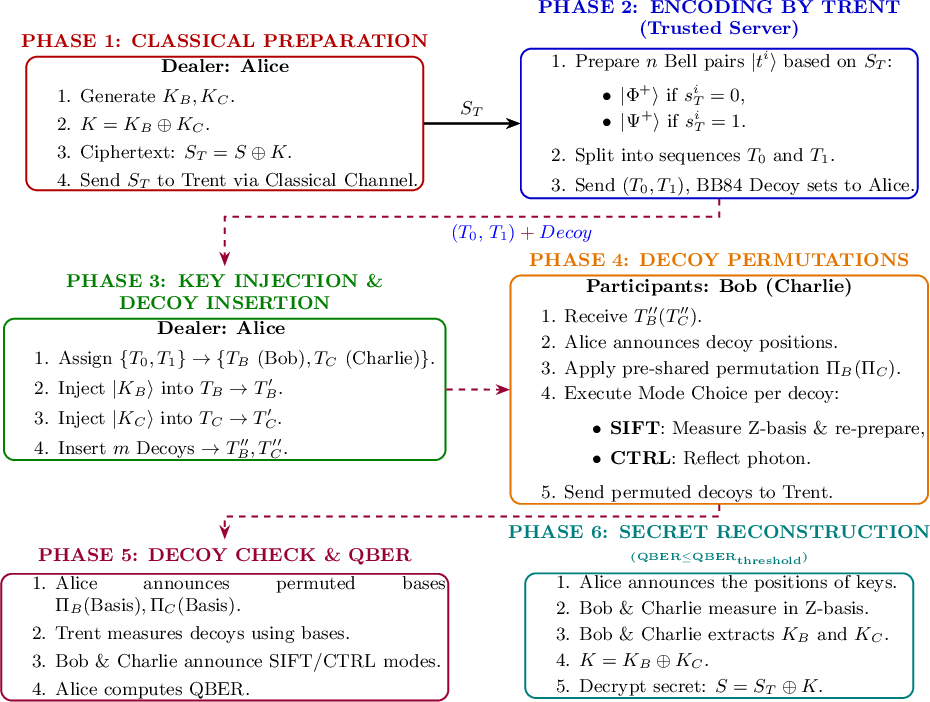}
\caption{Overall architecture of the proposed SQSS protocol. The dotted arrows denote quantum communication, while the solid black arrows denote classical communication among the participating entities.}
\label{fig:qss_protocol}
\end{figure}

% ============Protocol ends here===========

 % ==================================================
% \begin{figure}
%     \centering
%   \includegraphics[width = \textwidth]{SQSS_protocol.png}
% \end{figure}

\section{Security Analysis}
We now proceed to demonstrate that the proposed SQSS protocol is secure against eavesdropping attacks. To establish the comprehensive security of the proposed protocol, we first evaluate its resistance against the most severe internal threat model: the collusion between the semi-quantum participant Bob (which can consider any one party) and the quantum third-party $\mathrm{Trent}$ (denoted as the coalition $\mathcal{E}_{\{B, T\}}$). In QSS frameworks, internal participants possess structural advantages over external eavesdroppers (Eve), such as direct access to quantum channels, knowledge of operational modes, and execution authority over intermediate measurements. Consequently, proving security against a colluding coalition of internal entities represents the tightest upper bound on eavesdropping capabilities~\cite{10.3389/fphy.2022.1021113,deng2006circular}.

\begin{theorem}
    The proposed SQSS protocol is information-theoretically secure against the maximal internal collusion between the semi-quantum participant Bob $(B)$ and the quantum third-party $\mathrm{Trent}$ $(T)$.
\end{theorem}

\begin{proof}
    \textit{1. Information-Theoretic Blindness of Ciphertext $S_T$: } 
    
    Let $S \in \{0, 1\}^n$ be the secret message, $K_B, K_C \xleftarrow{\$} \{0, 1\}^n$ be independent, uniformly distributed random key strings, and $K = K_B \oplus K_C$ be the master key.The ciphertext delivered to $\mathrm{Trent}$ is $S_T = S \oplus K$. The conditional Shannon entropy of the secret $S$ given the ciphertext $S_T$ and Bob's key $K_B$ is evaluated as
    
    $$H(S \mid S_T, K_B) = H(S \mid S \oplus K_B \oplus K_C, K_B).$$
    
    Since $K_C$ is selected uniformly at random from $\{0, 1\}^n$ by Alice and kept private, $K_B \oplus K_C$ remains a perfectly uniform random variable over $\{0, 1\}^n$, even when $K_B$ is known
    
    $$P(K_B \oplus K_C = k \mid K_B = k_b) = \frac{1}{2^n}, \quad \forall k \in \{0, 1\}^n.$$
    
    Thus, by the properties of the One-Time Pad,
    $$H(S \mid S_T, K_B) = H(S).$$
    
    This proves that the colluding coalition $\mathcal{E}_{\{B, T\}}$ gains zero mutual information about $S$ from $S_T$ unless they can extract Charlie's private key $K_C$
    $$I(S ; S_T, K_B) = H(S) - H(S \mid S_T, K_B) = 0.$$

    \textit{2. Eavesdropping Model on Charlie's Key ($K_C$): }
    
    To extract $K_C$, the coalition $\mathcal{E}_{\{B, T\}}$ must measure or intercept the quantum stream $T''_C$ transmitted from Alice to Charlie during Phase 3, or the permuted sequence returned by Charlie during Phase 4. Let the quantum sequence sent to Charlie be decomposed into disjoint position index sets:$$\mathbb{I}_{\text{total}} = \mathbb{I}_{\text{decoy}} \cup \mathbb{I}_{\text{entangle}} \cup \mathbb{I}_{\text{key}, C}$$where $\vert{} \mathbb{I}_{\text{key}, C}\vert{} = n$, $\vert{}\mathbb{I}_{\text{decoy}}\vert{} = m$, and $\vert{} \mathbb{I}_{\text{total}}\vert{} = N = 2n + m$.
\medskip

    \textit{3. Undistinguishable Key Slots and Detection Probability: }
    
    During the quantum transmission phases (Phases 3 and 4), the exact position sets $\mathbb{I}_{\text{key}, C}$ and $\mathbb{I}_{\text{decoy}}$ are known exclusively to Alice. Alice reveals $\mathbb{I}_{\text{key}, C}$ only in Step 6.1 after the QBER check succeeds. Suppose $\mathrm{Trent}$ or Bob, attempts an intercept-measure-resend attack on $k$ arbitrary quantum positions in $T''_C$ to guess $K_C$. Then 

    \begin{itemize}
    \item{\textit{Probability of Hitting a Decoy Slot}:}
    
    Since indices are uniformly randomized, the probability that a single intercepted particle belongs to the decoy set $I_{\text{decoy}}$ is:    $$P_{\text{decoy}} = \frac{m}{N}.$$

    \item \textit{Error Generation on $X$-Basis Decoys}:
    
    Decoy states prepared in the diagonal basis $X = \{\vert{}+\rangle, \vert{}-\rangle\}$ collapse to computational basis states $\{\vert{}0\rangle, \vert{}1\rangle\}$ under any $Z$-basis measurement by $\mathcal{E}_{\{B, T\}}$. If Charlie chooses $\mathrm{CTRL}$ mode (direct reflection) on an $X$-basis decoy particle (which occurs with probability $\Pr(\text{$\mathrm{CTRL}$}) \approx \frac{1}{2}$), the expected fidelity $F$ of the returned state measured in $X$-basis by $\mathrm{Trent}$ in Phase 5 is $F =  \frac{1}{2}.$
    This induces an error probability of $e = \frac{1}{2}$ for every measured $X$-basis decoy. 
    \medskip

    \item \textit{Cumulative Detection Probability \emph{($\Pr(\text{detect})$):}}
    For $m_X \approx \frac{m}{2}$ decoy particles prepared in the $X$-basis, if the coalition intercepts the entire sequence ($k = N$), the probability of evading detection across all $\mathrm{CTRL}$-mode $X$-basis decoys is bounded by
    $$\Pr(\text{escape}) = \left( 1 - \Pr(\text{$\mathrm{CTRL}$}) \cdot e \right)^{m_X} = \left( 1 - \frac{1}{2} \cdot \frac{1}{2} \right)^{\frac{m}{2}} = \left( \frac{3}{4} \right)^{\frac{m}{2}}.$$
    
    Therefore, the probability of detecting the colluding coalition $\mathcal{E}_{\{B, T\}}$ approaches unity exponentially as $m$ increases $$\Pr(\text{detect}) = 1 - \Pr(\text{escape}) = 1 - \left( \frac{3}{4} \right)^{\frac{m}{2}} \xrightarrow{m \to \infty} 1.$$
    Hence, detection is guaranteed as $m$ increases.    
\medskip

     \item \textit{Security Under Secret Local Permutation ($\Pi_C$):}
     Before Charlie returns his decoy sequence to $\mathrm{Trent}$ in Phase 4, he applies his secret local permutation $\Pi_C$:$$\widetilde{T}''_{C,\text{Decoy}} = \Pi_C(T''_{C, \text{Decoy}})$$Since $\Pi_C$ is a bijection known only to Alice and Charlie ($\Pi_C \in S_{N}$ where $S_{N}$ is the symmetric group of degree $N$), $\mathrm{Trent}$ receives a randomized sequence ordering. Even if $\mathrm{Trent}$ measures all returning particles, he cannot align any measurement outcome with its original position slot without knowledge of $\Pi_C^{-1}$. Any attempt by $\mathrm{Trent}$ to reorder the particles introduces systematic positional mismatches, driving the calculated QBER high enough to abort the protocol.

\end{itemize}

Because $I(S ; S_T, K_B) = 0$ and any active attempt by $\mathcal{E}_{\{B, T\}}$ to intercept $K_C$ is detected with probability $\Pr({\text{detect}}) = 1 - (\frac{3}{4})^{\frac{m}{2}}$, the coalition $\mathcal{E}_{\{B, T\}}$ cannot extract the secret $S$ without doing the abort in Phase 5 of the proposed protocol.  Hence, the claimed result follows. \qed
\end{proof}

As the primary objective of our attack is to consider the internal collusion between Bob and Trent, we must also consider this scenario here. According to the protocol steps described above, even if Bob and Trent collude, they cannot obtain any information about Charlie’s secret permutation $\Pi_C$. Consequently, they cannot gain any information about Charlie’s secret decoy sequence, since the permutation $\Pi_C$ remains unknown to both parties.

\subsection{Double CNOT Attacks (DCNA)}

In the following security analysis, we assume that the adversary may be either Bob, Charlie, or any external eavesdropper. We collectively denote such an adversary by Eve. Without loss of generality, we consider the case where Eve applies a CNOT operation to the quantum particles transmitted from Alice to Bob, and a second CNOT is applied when the particles travel from Bob to Trent. Similar situation for the participant Charlie also.

Let an adversary Eve attempt to extract computational $Z$-basis key information from a traveling quantum sequence $T''_B$ using an auxiliary probe qubit $\ket{E}_E$ initialized to $\ket{0}_E$. Eve applies a Unitary Controlled-NOT operation $\hat{U}_{\text{CNOT}}$, where Bob's (Charlie's) particle serves as the control qubit and the probe serves as the target qubit
\begin{equation*}
\hat{U}_{\text{CNOT}} \ket{x}_{B(C)} \ket{0}_{E_{1(2)}} = \ket{x}_{B(C) }\ket{x}_{E_{1(2)}}, \quad \forall x \in \{0, 1\}
\end{equation*}

Now we show that the proposed SQSS protocol is secure against such Double CNOT Attacks (DCNA).

% \begin{theorem}[Security Against Double CNOT Attack]
% \label{thm:dcna_attack}
% The proposed semi-quantum secret sharing (SQSS) protocol is unconditionally secure against Double CNOT Attacks (DCNA). When an adversary applies two sequential CNOT operations-one on the forward path and one on the return path—the attack induces a deterministic $50\%$ error rate on $X$-basis decoys in $\mathrm{CTRL}$ mode and a $50\%$ error rate in $\mathrm{SIFT}$ mode, yielding an expected error count of $\mathbb{E}[X_{\text{err}}] = m/2$ across $m$ diagonal decoy particles and a detection probability $P_{\text{detect}} = 1 - 2^{-m}$.
% \end{theorem}
% \medskip

% \begin{proof}

% \noindent Let an adversary Eve initialize two auxiliary probe qubits $\ket{E_1}_{E_1}$ and $\ket{E_2}_{E_2}$ to state $\ket{00}_{E_1 E_2}$.
Adversary Eve executes the Double CNOT attack as follows:
\begin{itemize}
    \item \textbf{Forward CNOT ($\hat{U}_{\mathrm{CNOT}}^{(f)}$):} Applied on traveling particle $B(C)$ and probe $E_{1(2)}$ during the forward transmission [Alice $\to$ Bob (Charlie)].
    \item \textbf{Bob (Charlie)'s Operation:} Bob (Charlie) applies either $\mathrm{CTRL}$ mode (direct reflection) or $\mathrm{SIFT}$ mode ($Z$-basis measurement and re-preparation).
    \item \textbf{Backward CNOT ($\hat{U}_{\text{CNOT}}^{(b)}$):} Applied on returning particle $B (C)$ and the same probe $E_{1(2)}$ from forward case, during the backward transmission [Bob (Charlie) $\to$ Trent].
\end{itemize}

\noindent\textbf{1. Computational Basis ($Z$-Basis Decoy States):} For the decoy particles prepared in the computational basis $\{\ket{0}_*, \ket{1}_*\}$, the transformation yields:
    \begin{equation*}
    \hat{U}_{\text{CNOT}} \ket{0}_* \ket{0}_{E_{1(2)}} = \ket{0}_* \ket{0}_{E_{1(2)}}, \quad \hat{U}_{\text{CNOT}} \ket{1}_* \ket{0}_{E_{1(2)}}= \ket{1}_* \ket{1}_{E_1(2)}.
    \end{equation*}
    The probe extracts the classical bit value without inducing any changes on the main qubit. These bits are easily copied using a CNOT attack. However, as Eve does not know the information about the position of such decoy states, he gains no information after the first CNOT operation.

\noindent\textbf{2. Forward Pass Transformation:}
For a particle prepared in the diagonal basis state $\ket{\psi}_*\in \{\ket{+}, \ket{-}\}$, the first CNOT operation produces the composite state
\begin{equation*}
\hat{U}_{\text{CNOT}}^{(f)} \ket{\pm}_*\ket{0}_{E_{1(2)}} = \frac{1}{\sqrt{2}} \left( \ket{0}_*\ket{0}_{E_{1(2)}} \pm \ket{1}_*\ket{1}_{E_{1(2)}} \right) = \ket{\Phi^\pm}_{*E_{1(2)}}.
\end{equation*}

\noindent\textbf{3. Evaluation Across Receiver Operational Modes:}\\
\vskip-0pt
\noindent\textbf{\emph{Case A: Bob (Charlie) operates in $\mathrm{SIFT}$ Mode ($Z$-Basis Measurement)}}

\noindent In $\mathrm{SIFT}$ mode, Bob (Charlie) measures particle $B (C)$ in the $Z$-basis after the first CNOT operation:
\begin{itemize}
    \item With probability $1/2$, Bob (Charlie) measures $\ket{0}_{B(C)}$. System collapses to $\ket{0}_{B(C)} \ket{0}_{E_{1(2)}}$. Bob (Charlie) re-prepares and returns $\ket{0}_{B(C)}$.
    \item With probability $1/2$, Bob (Charlie) measures $\ket{1}_{B(C)}$. System collapses to $\ket{1}_{B(C)} \ket{1}_{E_{1(2)}}$. Bob (Charlie) re-prepares and returns $\ket{1}_{B(C)}$.
\end{itemize}

% Applying the second CNOT operation $\hat{U}_{\text{CNOT}}^{(2)}$ on the returned particle $B$ and probe $E_2$:
% \begin{equation*}
% \hat{U}_{\text{CNOT}}^{(2)} \ket{0}_B \ket{0}_{E_2} = \ket{0}_B \ket{0}_{E_2}, \quad \hat{U}_{\text{CNOT}}^{(2)} \ket{1}_B \ket{0}_{E_2} = \ket{1}_B \ket{1}_{E_2} .
% \end{equation*}
As the initial decoy states are $\ket{+}$ or $\ket{-}$, upon applying the first $\mathrm{CNOT}$ gate targeting Eve's probe state $\ket{0}_E$, the joint state evolves into the maximally entangled Bell state $\ket{\Phi^+}$ or $\ket{\Phi^-}$, respectively. After Eve applies the second $\mathrm{CNOT}$ gate, the joint system disentangles, reverting to the unentangled product states $\ket{+}_B \ket{0}_{E_1}$ and $\ket{-}_B \ket{0}_{E_1}$ (similar for Charlie also). Consequently, no residual entanglement remains between Bob's (Charlie's) decoy particle and Eve's probe state. When Bob (Charlie) performs a $Z$-basis measurement, the outcome is completely uncorrelated with Eve's intervention and perfectly mirrors the reference measurement results expected by Trent. As a result, eavesdropping cannot be detected in this case.

% and the states before applying second CNOT is $\ket{\Phi^{\pm}}$, so again applying a second CNOT it gives either $\ket{0}_{B(C)} \ket{0}_{E_{1(2)}}$ or $\ket{1}_{B(C)} \ket{1}_{E_{1(2)}}$. The states of the particle become independent of each other and measurement in $Z$ basis cannot keep any entanglement between the orginial state and probe state, also Bob (Charlie) has the same measurement result as of Trent. So in this case eavesdropping detection is not possible.

\paragraph{\textbf{Case B: Bob (Charlie) operates in $\mathrm{CTRL}$ Mode (Direct Reflection)}}
In $\mathrm{CTRL}$ mode, Bob (Charlie) reflects his particle back without measurement. However, as the reordering is unknown to Eve of the decoy particles, he cannot apply the second CNOT operation $\hat{U}_{\text{CNOT}}^{(b)}$ on the same decoy particle $*$. Rather, he ends up by applying $\hat{U}_{\text{CNOT}}^{(b)}$ on a different particle $\circ$, which is either from $\{\ket0_\circ,\ket1_\circ\}$ or one particle from the states $\ket{\Phi^\pm}$. Then the combine state becomes
% \begin{align*}
% \ket{\Upsilon_1}&:=
% \hat V^{(1)}\ket{\Phi^\pm}_{*E_{1(2)}}\ket{\alpha}_{\circ}, &\mathrm{\ if\ }\ket{\alpha}_\circ\in\{\ket{0}, \ket{1}\}\\
% \ket{\Upsilon_2}&:= \hat{V}^{(2)}\ket{\Phi^\pm}_{*E_{1(2)}}\ket{\Phi^\pm}_{\circ E'_{1(2)}}, &\mathrm{\ otherwise,}
% %=\ket{\Upsilon^\pm}_{E_{1(2)}*\circ},
% \end{align*}

\begin{align*}
\ket{\Upsilon_1}
&:= \hat{V}^{(1)}\ket{\Phi^\pm}_{*E_{1(2)}}\ket{\alpha}_{\circ},
\quad \text{when } \ket{\alpha}_{\circ}\in\{\ket{0},\ket{1}\},\\
\ket{\Upsilon_2}
&:= \hat{V}^{(2)}\ket{\Phi^\pm}_{*E_{1(2)}}
   \ket{\Phi^\pm}_{\circ E'_{1(2)}},
\quad \text{otherwise},
\end{align*}
where $\hat{V}^{(1)}:=\hat{I}_2\otimes\hat{U}_{\text{CNOT}}^{(b)}$ and $\hat{V}^{(2)}:=\hat{I}_2\otimes\hat{U}_{\text{CNOT}}^{(b)}\otimes\hat{I}_2$.
% where 
%$\ket{\Upsilon^\pm}$ is a three qubit GHZ-like state for 
% , or
% \begin{align*}
% .=\ket{\Upsilon^\pm}_{E_{1(2)}*\circ},
% \end{align*}% otherwise, for $\ket{\alpha} \in \{\ket \}$ it is a part of a larger state when $\ket{\alpha}_\circ$ represents one qubit from another $\ket{\Phi^\pm}$, generated by the forward CNOT operation applied on a different traveling qubit $\circ$ with a different probe $E'_{1(2)}$.

In this situation, $\mathrm{Trent}$ measures the decoy particle $*$ in $X$-basis getting $0$ or $1$, with probability $\frac{1}{2}$ each, as
\begin{align}
H\otimes\hat{I}\ket{\Upsilon_1}&=\hat{V}^{(1)}(H\otimes\hat I_4)\ket{\Phi^\pm}_{*E_{1(2)}}\ket{\alpha}_{\circ} \notag \\
&=\hat V^{(1)}\frac{1}{\sqrt2}(\ket0_*\ket\pm_{E_{1(2)}}+\ket1_*\ket\mp)_{E_{1(2)}} \ket{\alpha}_{\circ}  \\
  \notag  \\
H\otimes\hat{I}\ket{\Upsilon_2}&=\hat{V}^{(2)}(H\otimes\hat I_8)\ket{\Phi^\pm}_{*E_{1(2)}} \ket{\Phi^{\pm}}_{\circ E'_{1(2)}} \notag \\
&=\hat V^{(2)}\frac{1}{\sqrt2}(\ket0_*\ket\pm_{E_{1(2)}}+\ket1_*\ket\mp)_{E_{1(2)}} \ket{\Phi^{\pm}}_{\circ E'_{1(2)}} 
\end{align}

\noindent and the CNOT operation in $\hat V^{(i)}$ (for $i =1,2$) has no effect on the particle $*$ .

% Taking the partial trace over both probe systems $E_1$ and $E_2$:
% \begin{equation*}
% \rho_B^{\text{$\mathrm{CTRL}$}} = \text{Tr}_{E_1, E_2} \left( \ket{\Psi_{\text{$\mathrm{CTRL}$}}^+}\bra{\Psi_{\text{$\mathrm{CTRL}$}}^+} \right) = \frac{1}{2}\ket{0}\bra{0} + \frac{1}{2}\ket{1}\bra{1} = \frac{1}{2}\hat{I}_2 .
% \end{equation*}

% Since $\rho_B^{\text{$\mathrm{CTRL}$}} = \frac{1}{2}\hat{I}_2$ is maximally mixed, measuring particle $B$ in the original $X$-basis during Phase 5 yields:
% \begin{equation*}
% \Pr(\text{Error} \mid \text{$\mathrm{CTRL}$}, X\text{-decoy}) = \text{Tr}\left( \ket{\mp}\bra{\mp} \cdot \rho_B^{\text{$\mathrm{CTRL}$}} \right) = \frac{1}{2} .
% \end{equation*}
% In both branches, the returned state for particle $B$ remains a $Z$-basis eigenstate ($\ket{0}_B$ or $\ket{1}_B$). When $\mathrm{Trent}$ measures this returned state in the original $X$-basis during Phase 5:
% \begin{equation*}
% \Pr(\text{Error} \mid \text{$\mathrm{SIFT}$}, X\text{-decoy}) = |\braket{\mp|0}|^2 = |\braket{\mp|1}|^2 = \frac{1}{2}. 
% \end{equation*}

\subsubsection*{3. Error Quantification and Detection Probability}
% For any $X$-basis decoy particle, DCNA induces an error probability $e_{\text{DCNA}} = \frac12$ across both operational modes. 

If there are $m_X$ diagonal decoy particles out of $n$ in the transmitted particle sequence $T''_B$ (for Bob) total transmitted particles, then the probability of the DCNA attack being detected is 

\begin{align*}
  \Pr(\mathrm{Detection}) &= 1 - \Pr(\mathrm{No \; Detection}) \\
    \Pr(\mathrm{Detection}) &= 1 - \left[\frac{1}{n + m_X} \times 1 \times \frac12 + \left(1 - \frac{1}{n + m_X} \right)\times \frac12 \times \frac12 \right]^{m_X} \\ 
    &= 1 - \left[ \frac14 + \frac{1}{4(n +m_X)}\right]^{m_X}.
\end{align*}
The event ``No Detection'' happens surely when $\mathrm{SIFT}$ operation is being chosen and when $\mathrm{CTRL}$ is chosen this event happens with probability $\tfrac12$.
As Bob (or Charlie) chooses between the $\mathrm{CTRL}$ and $\mathrm{SIFT}$ operations uniformly at random with equal probability $\frac{1}{2}$, each corresponding term in the ``No Detection'' expression incorporates a scaling factor of $\frac{1}{2}$.

 As the number of diagonal decoy particles increases, i.e., ($m_X \rightarrow \infty$), the probability of detecting the eavesdropping approaches $1$. Therefore, by choosing a sufficiently large, precomputed value of ($m_X$), the DCNA attack can be detected with an arbitrarily high probability.
% total detected errors $X_{\text{err}}$ follow:
% \begin{equation*}
% X_{\text{err}} \sim \text{Binomial}\left(m_X, \frac{1}{2}\right), \quad \mathbb{E}[X_{\text{err}}] = \frac{m_X}{2} .
% \end{equation*}
% The overall sequence error rate is $\text{QBER}_{\text{total}} = \frac{m_X}{2n}$. The probability of evading detection ($P_{\text{evade}}$) and the corresponding detection probability ($P_{\text{detect}}$) are:
% \begin{equation*}
% P_{\text{evade}} = 2^{-m_{X}}, \quad P_{\text{detect}} = 1 - 2^{-m_{X}} .
% \end{equation*}
% As $m_X$ tends to large values, $P_{\text{evade}} \to 0$ exponentially, making DCNA detectable with probability arbitrarily close to 1.
% \qed
% \end{proof}

\subsection{Entangle-Measure Attack}

Detection of an Entangle-Measure (EM) attack strictly requires the inclusion of $X$-basis decoy states ($\{\ket{+}, \ket{-}\}$), as probe-induced phase decoherence remains completely invisible to computational $Z$-basis measurements. When an adversary entangles an auxiliary probe to gain key information, the phase coherence of traveling particles is unavoidably destroyed, forcing a $50\%$ error rate on reflected $X$-basis decoys during $\text{$\mathrm{CTRL}$}$ mode. Consequently, evaluating our protocol under the $X$-basis scenario is essential to prove that any eavesdropping attempt drives the evasion probability exponentially to zero.
\medskip 

    Let an eavesdropper or dishonest quantum party Eve initialize a probe state $\ket{E}_E = \ket{0}_E$ and execute a general entangling operation $\hat{U}_E$ on passing particle $B$, as follows
\begin{align*}
\hat{U}_E \ket{0}_B \ket{0}_E &= \alpha \ket{0}_B \ket{e_{00}}_E + \beta \ket{1}_B \ket{e_{01}}_E \; ,\\
\hat{U}_E \ket{1}_B \ket{0}_E &= \gamma \ket{0}_B \ket{e_{10}}_E + \delta \ket{1}_B \ket{e_{11}}_E \; ,
\end{align*}
    where $|\alpha|^2 + |\beta|^2 = 1$, $|\gamma|^2 + |\delta|^2 = 1$, and $\ket{e_{ij}}_E$ represent normalized probe states.

    To prevent immediate detection during $Z$-basis decoy checks in $\mathrm{SIFT}$ mode, Eve must avoid bit-flip errors by setting $\beta = \gamma = 0$, yielding $\hat{U}_E \ket{0}\ket{0} = \alpha \ket{0}\ket{e_{00}}$ and $\hat{U}_E \ket{1}\ket{0} = \delta \ket{1}\ket{e_{11}}$ with $|\alpha| = |\delta| = 1$.

Now consider a diagonal decoy particle prepared in state $\ket{+}_B = \frac{1}{\sqrt{2}}(\ket{0}_B + \ket{1}_B)$. The action of $\hat{U}_E$ creates the composite state:
\begin{equation*}
\ket{\Psi_{BE}} = \hat{U}_E \ket{+}_B \ket{0}_E = \frac{1}{\sqrt{2}} \left( \alpha \ket{0}_B \ket{e_{00}}_E + \delta \ket{1}_B \ket{e_{11}}_E \right).
\end{equation*}

    In $\mathrm{CTRL}$ mode, Bob reflects particle $B$ without measurement or disturbance. Upon return, the state is evaluated in the diagonal $X$-basis ($\{\ket{+}, \ket{-}\}$). The probability of measuring the correct state $\ket{+}_B$ is:
\begin{align*}
\Pr(\text{Correct} \mid \text{$\mathrm{CTRL}$}, \ket{+}) &= \tr\left( \ket{+}\bra{+}_B \cdot \tr_E(\ket{\Psi_{BE}}\bra{\Psi_{BE}}) \right) \\
    &= \frac{1}{4} \big( 2 + 2 \text{Re}(\alpha^* \delta \bra{e_{00}}\ket{e_{11}}) \big).
\end{align*}

    Defining $F = \text{Re}(\alpha^* \delta \bra{e_{00}}\ket{e_{11}})$ as the probe fidelity overlap parameter ($0 \le F \le 1$), the corresponding error probability $e_{\text{EM}}$ on $X$-basis decoys is
\begin{equation*}
e_{\text{EM}} = 1 - \Pr(\text{Correct}) = \frac{1 - F}{2}.
\end{equation*}

    If Eve attempts to gain information regarding the key sequence, the probe states must be orthogonal ($\bra{e_{00}}\ket{e_{11}} = 0 \implies F = 0$), forcing $e_{\text{EM}} = \tfrac12$. Because participants withhold operation mode announcements ($\text{$\mathrm{SIFT}$}/\text{$\mathrm{CTRL}$}$) until the completion of particle transmission, Eve cannot selectively avoid $\mathrm{CTRL}$ mode decoys, guaranteeing detection via $\text{QBER}$ threshold evaluation.

\subsection{Decoy State Substitution Attack}

All of the attacks we discussed above, has been detected due to the presence of the diagonal decoy states $(\ket\pm)$. Suppose a dishonest Quantum Third Party ($\text{QTP}$) Trent executes a Decoy State Substitution Attack (DSSA) by maliciously replacing the sequences $X^+$ and $X^-$ of $X$-basis decoy states ($\ket{\pm}$) with $X^+_s$ and $X^-_s$, respectively, containing computational $Z$-basis states ($\ket{0}$ and $\ket{1}$) during the decoy preparation in Step 2.4. %and present them to Alice as sequences of $\ket+$ and $\ket-$, respectively. 
We assume that $\mathrm{Trent}$ sends the sequences of $Z$-basis decoy states correctly as $Z^0$ and $Z^1$. The above substitution can be performed in two ways, as follows.\\
\textbf{Case A [Each sequence contains only one state]: }Without loss of generality, assume that $X^+_s$ ($X^-_s$) consists entirely of the state $\ket{0}$ ($\ket{1}$). During the $\mathrm{SIFT}$ mode, Bob and Charlie would therefore obtain the same outcome with probability $1$ for all states in each sequence. In contrast, for the actual sequences consisting of $\ket{+}$ and $\ket{-}$ states, Bob and Charlie should obtain $\ket{0}$ and $\ket{1}$ with equal probability $\frac{1}{2}$. Thus, such a substitution can be detected.

In the $\mathrm{CTRL}$ mode, however, measuring the state in the $Z$-basis allows Trent to determine which sequence the state belongs to. By announcing the corresponding result accordingly, Trent may avoid detection.\\
\textbf{Case B [Both sequences contain both $\ket0$ and $\ket1$]: }Because a classical participant measuring strictly in the $Z$-basis% observes a uniform probability distribution $|\braket{0|\pm}|^2 = |\braket{1|\pm}|^2 = 1/2$
, these substituted states cannot be detected during $\mathrm{SIFT}$ mode. 

However, the attack fails during $\mathrm{CTRL}$ mode. For a genuine diagonal decoy state $\ket{\psi} \in \{\ket{+}, \ket{-}\}$, an undisturbed reflection in $\mathrm{CTRL}$ mode yields a deterministic measurement outcome matching the prepared state when measured in the $X$-basis
\begin{equation*}
\Pr(\text{Correct Outcome} \mid \text{Genuine } X\text{-decoy}) = |\braket{\psi|\psi}|^2 = 1.
\end{equation*}

Conversely, if Trent substituted a computational state $\ket{\phi} \in \{\ket{0}, \ket{1}\}$, measuring the reflected state in the $X$-basis yields an incorrect decoy outcome with a probability of
\begin{equation*}
\Pr(\text{Error} \mid \text{Substituted } Z\text{-state}) = |\braket{\psi|\phi}|^2 = \frac{1}{2}.% = 50\%
\end{equation*}
Even if Trent does not perform an $X$-basis measurement, since he knows that no decoy state is prepared in the $X$-basis, he still cannot determine whether the substituted state is from $X^+_s$ or $X^-_s$. Consequently, he must guess the sequence, which is correct with probability $\frac{1}{2}$ only.

Crucially, participants defer the announcement of their operational modes ($\text{$\mathrm{SIFT}$}$ or $\text{$\mathrm{CTRL}$}$) until Trent has completed all measurements on the returned sequence and declare the outcomes. Because Trent cannot predict which particle undergoes which mode, the error induced in each cases generates a strong statistical anomaly during public correlation checks, guaranteeing immediate detection of dishonest $\text{QTP}$ and secure termination of the protocol.

\section{Qubit Efficiency Analysis}
The qubit efficiency of a quantum communication protocol is defined according to Cabello's model~\cite{PhysRevLett.85.5635} as
$\eta = \frac{\delta_1}{\delta_2}$,
where $\delta_1$ represents the length of the transmitted classical secret sequence, and $\delta_2$ denotes the total number of qubits prepared across all stages of the protocol.

In our updated semi-quantum secret sharing (SQSS) protocol:
\begin{enumerate}
    \item \textbf{Secret Payload ($\delta_1$):} Alice transmits an $n$-bit secret $S$, yielding $\delta_1 = n$.
    \item \textbf{Total Prepared Qubits ($\delta_2$):} 
    \begin{itemize}
        \item Trent prepares $n$ Bell states, comprising $2n$ single qubits ($T_0$ and $T_1$).
        \item Alice prepares $2n$ single-qubit key states ($|K_B\rangle$ and $|K_C\rangle$) for quantum key injection.
        \item Alice inserts $2m$ decoy state particles ($m$ each into $T'_B$ and $T'_C$).
    \end{itemize}
\end{enumerate}

% Neglecting re-prepared particles during $\mathrm{SIFT}$ mode
Therefore, the base prepared count is $\delta_2 = 4n + 2m$. Consequently, the qubit efficiency $\eta$ of the updated scheme is:
\begin{equation*}
\eta = \frac{n}{4n + 2m} = \frac{1}{4 + 2\left(\frac{m}{n}\right)} = \frac{1}{4+ 2\gamma} \; ,
\end{equation*}

where $\gamma = \frac mn$ the ratio of decoy particles relative to payload particles.

\subsection{Comparison }
In this section, we compare the scheme that we have proposed with a number of well-known quantum secret sharing protocols by examining various performance indicators and security measures. The evaluation centres on the operational complexity of the quantum resources needed (for example, multi-qubit product states, GHZ states, or standard Bell states) and on the hardware capabilities of the secret dealer, making a clear distinction between schemes which require a fully quantum dealer and those which allow for a classical or semi-quantum party. We also look at the protocol's security resilience, especially in relation to eavesdropping attacks such as the Double CNOT Attack (DCNA), the type of secrets that are shared (whether specific target bits or arbitrary bit strings are set up), and the qubit efficiency, which is the ratio of secret bits obtained per quantum state transmitted.

% ========================================
% \renewcommand{\thefootnote}{\fnsymbol{footnote}}
\renewcommand{\thefootnote}{\dag}
% =========================================

\begin{table}[htbp]
\centering
\caption{\label{Tab:Comparisons}Comparison of the Proposed and Existing SQSS Protocols}
\begin{tabular}{||lllclc||} 
\hline
\multirow{3}{*}{\textbf{Protocols \quad}} & 
\multirow{2}{*}{\textbf{Quantum \quad}} & \textbf{Dealer's } & \multicolumn{1}{l}{\textbf{Security}} & \multicolumn{1}{c}{\multirow{2}{*}{\textbf{Shared}}} & \multicolumn{1}{l||}{\multirow{2}{*}{\textbf{Qubit}}}\\
 & \multirow{2}{*}{\textbf{resources }} & \textbf{secret} & \multicolumn{1}{l}{\textbf{against}} & \multicolumn{1}{c}{\multirow{2}{*}{\textbf{secrets}}} & \multicolumn{1}{l||}{\multirow{2}{*}{\textbf{efficiency}}}\\
 & & \textbf{state } & \multicolumn{1}{l}{\textbf{DCNA}} & & \\
\hline \hline
& $n$-qubit&&&&\\
Li \emph{et al.}~\cite{Li2018} & product & Quantum & No & {Unspecific bits} & $\frac14$ \\
&states&&&&\\
\hline
Hou \emph{et al.}~\cite{Hou2024Circular} & GHZ state & Quantum & Yes & Unspecific bits & $\frac13$ \\
\hline
Tian \emph{et al.}~\cite{Tian2021} & Bell state & Quantum & No & Specific bits & $\frac14$ \\
\hline
Ma \emph{et al.}~\cite{Ma_2024} & Bell state & Classical & No & Specific bits & $\frac{1}{[2(n+c)]}$ \\[2pt]
\hline
He \emph{et al.}~\cite{He2026} & Bell state & Classical & Yes & Specific bits & $\frac{1}{4+2\gamma}{\footnotemark}$\\
\hline
\textbf{Our}& \multirow{2}{*}{\text{Bell state}} & \multirow{2}{*}{\text{Classical}} & \multirow{2}{*}{\text{Yes}} & \multirow{2}{*}{\text{Specific bits}} & \multirow{2}{*}{${\frac{1}{4+ 2\boldsymbol\gamma}}$} \\
\textbf{Protocol}&&&&&\\
\hline
\end{tabular}
\end{table}
\vskip-2pt

The comparative analysis shown in Table~\ref{Tab:Comparisons} demonstrates that, with the equivalent parameter values, our proposed protocol successfully eliminates collusion vulnerabilities between classical participants and untrusted quantum third parties, thereby establishing a robust defense against the collusion between untrusted classical users and quantum third parties, while ensuring the quantum particle consumption remains strictly optimal.

\footnotetext{While the authors in~\cite{He2026} calculate a qubit efficiency of $\frac{n}{4n} = \frac{1}{4}$ by omitting the decoy states, but if  $m$ decoy particles are inserted by each party then it expands the total quantum resource to $4n + 2m$, yielding a revised efficiency of $\eta = \frac{1}{4 + 2 {\frac{m}{n}}}$, which is same as our protocol .}

%=================================
\section{\label{sec:cd}Conclusion and Future Direction}
In this paper, we identified and analyzed a critical security vulnerability in the SQSS protocol proposed by He \emph{et al.}~\cite{He2026}, showing its susceptibility to collusion attacks involving both internal and external adversaries. To resolve these flaws, we presented an enhanced SQSS scheme that strictly preserves the operational definition of semi-quantumness, also establishing a robust defense against the collusion of internal and external adversaries. By incorporating diagonal $X$-basis decoy states, secret permutation functions, and enforcing a delayed announcement of the operational mode ($\text{$\mathrm{SIFT}$}/\text{$\mathrm{CTRL}$}$), our protocol resolves the blind spots of classical $Z$-basis measurements against phase-probe attacks.

Looking ahead, several promising avenues remain to extend this research. A key extension involves generalizing the scheme from a two-receiver setting to an $N$-party semi quantum secret sharing network using multi-qubit GHZ or $W$ states while maintaining minimal semi-quantum hardware capabilities, and accounting for collusion attacks by internal and external adversaries.  Finally, exploring asymmetric semi-quantum models, in which receivers possess heterogeneous quantum capabilities, may offer a promising direction for designing secret sharing protocols that optimize resource allocation and performance in practical quantum communication networks. 

% ===============================
\appendix

\section{Correctness of the Collusion Attack in Sect.~\ref{sec:Attack_scenario}}

\subsubsection{Initial Parameters and Execution Setup: \label{appendix:A}}

To formally prove the correctness of the universal $Z$-sweep collusion attack against the protocol, we analyze the execution trace using the following explicit classical sequences provided by Alice
$$S=\{1,0,1,0,1\},\;
K_B=\{0,1,1,0,0\},\;
K_C=\{1,0,0,0,1\}.$$

\paragraph*{\textbf{Step I-1 \emph{(Classical Encryption):}}}
Alice establishes the master encryption key $K$ by computing the bitwise XOR of her classical key sequences
\[
K = K_B \oplus K_C = \{0\oplus1, 1\oplus0, 1\oplus0, 0\oplus0, 0\oplus1\} = \{1, 1, 1, 0, 1\}
\]
Using $K$, Alice encrypts the secret sequence $S$ into the ciphertext sequence $S_T$
\[
S_T = S \oplus K = \{1\oplus1, 0\oplus1, 1\oplus1, 0\oplus0, 1\oplus1\} = \{0, 1, 0, 0, 0\}
\]
The classical sequence $S_T$ is transmitted to $\mathrm{Trent}$.

\paragraph*{\textbf{Steps I-2 to I-4 \emph{(Stream Preparation and Injection):}}}
Based on $S_T = \{0, 1, 0, 0, 0\}$, $\mathrm{Trent}$ produces five entangled EPR pairs assigned to states $\ket{\Psi^+}_{BC}$ or $\ket{\Phi^+}_{BC}$ according to the protocol rules. The Trents insert their decoy states, and Alice subsequently inserts her key states $\ket{K_B}$ and $\ket{K_C}$ randomly in the $Z$-basis.

Based on the explicit structural mapping of the quantum channels, Alice distributes her secret using 5 informational EPR pairs, while interleaving key fragments directly within the streams. The exact 16-element quantum state sequences transmitted to the channels are defined as:
\begin{multline*}
    T'''_B = \{ \ket{+}_{B1}, \mathbf{\ket{0}}, \mathbf{\ket{1}}, \ket{\psi}_B, \ket{+}_{B2}, \ket{\phi}_B, \ket{-}_{B1}, \mathbf{\ket{1}}, \mathbf{\ket{0}}, \mathbf{\ket{0}}, \\
    \ket{\psi}_B, \ket{0}_{B1},  \ket{\psi}_B, \ket{1}_{B2}, \ket{\psi}_B, \ket{1}_{B2} \}
\end{multline*}
\begin{multline*}
    T'''_C = \{ \ket{\psi}_C, \ket{0}_{C2}, \ket{0}_{C1}, \ket{\phi}_C, \ket{+}_{C1}, \mathbf{\ket{1}}, \mathbf{\ket{0}}, \ket{\psi}_C, \ket{1}_{C2}, \\
    \ket{\psi}_C, \mathbf{\ket{0}}, \mathbf{\ket{0}}, \mathbf{\ket{0}}, \ket{1}_{C2}, \ket{\psi}_C, \mathbf{\ket{1}} \}
\end{multline*}
where bold entries highlight the classical key bits injected by Alice ($K_B$ and $K_C$), subscripts $B_i ,C_i$ denote Trent-inserted server decoys, and $\ket{\psi},\ket{\phi}$ represent the single particle shares from the informational Bell states $\ket{\Psi^+}_{BC},\ket{\Phi^+}_{BC}$.

\subsubsection{The Intercept-Resend Attack Simulation:} The malicious Bob (secretly colludes with $\mathrm{Trent}_1$ and $\mathrm{Trent}_2$) intercepts the transmission of $T'''_C$ to Charlie. Bob subjects all 16 slots to a universal $Z$-basis measurement operator $\mathcal{M}_z$.

\paragraph*{Complete Trace of Bob's Measurement Log: }
Table~\ref{tab:16_element_trace} demonstrates exactly how the universal $Z$-sweep captures the key data non-destructively while masking the modification of diagonal decoys through server collusion.

\begin{table}[htbp]
\centering
\caption{Universal $Z$-Sweep Trace on the 16-Element Sequence $T'''_C$}
\label{tab:16_element_trace}
\resizebox{\textwidth}{!}{
\begin{tabular}{ ||c ccccc|| }
\hline \rowcolor{headerblue}
\textbf{Slot ($j$)} & \textbf{ State} & \textbf{Classification} & \textbf{Measurement Mode} & \textbf{Logged Bit} & \textbf{State to Charlie} \\ \hline \hline
1  & $\ket{\psi}_C$    & Bell EPR Pair   & Deterministic ($100\%$)      & 0 or 1     & $\ket{0}$ or $\ket{1}$ \\
2  & $\ket{0}_{C2}$    & Rectilinear Decoy & Deterministic ($100\%$)    & 0          & $\ket{0}$ \\
3  & $\ket{0}_{C1}$    & Rectilinear Decoy & Deterministic ($100\%$)    & 0          & $\ket{0}$ \\
4  & $\ket{\phi}_C$    & Bell EPR Pair   & Deterministic ($100\%$)      & 0 or 1     & $\ket{0}$ or $\ket{1}$ \\
5  & $\ket{+}_{C1}$    & \textbf{Diagonal Decoy}  & \textbf{Basis mismatch} & \textbf{0 or 1} & $\mathbf{\ket{0}}$ \textbf{or} $\mathbf{\ket{1}}$ \\
6  & $\ket{1}$         & Alice Key ($K_C$) & Deterministic ($100\%$)     & 1          & $\ket{1}$ \\
7  & $\ket{0}$         & Alice Key ($K_C$) & Deterministic ($100\%$)     & 0          & $\ket{0}$ \\
8  & $\ket{\psi}_C$    & Bell EPR Pair   & Deterministic ($100\%$)      & 0 or 1     & $\ket{0}$ or $\ket{1}$ \\
9  & $\ket{1}_{C2}$    & Rectilinear Decoy & Deterministic ($100\%$)    & 1          & $\ket{1}$ \\
10 & $\ket{\psi}_C$    & Bell EPR Pair   & Deterministic ($100\%$)      & 0 or 1     & $\ket{0}$ or $\ket{1}$ \\
11 & $\ket{0}$         & Alice Key ($K_C$) & Deterministic ($100\%$)     & 0          & $\ket{0}$ \\
12 & $\ket{0}$         & Alice Key ($K_C$) & Deterministic ($100\%$)     & 0          & $\ket{0}$ \\
13 & $\ket{0}$         & Alice Key ($K_C$) & Deterministic ($100\%$)     & 0          & $\ket{0}$ \\
14 & $\ket{1}_{C2}$    & Rectilinear Decoy & Deterministic ($100\%$)    & 1          & $\ket{1}$ \\
15 & $\ket{\psi}_C$    & Bell EPR Pair   & Deterministic ($100\%$)      & 0 or 1     & $\ket{0}$ or $\ket{1}$ \\
16 & $\ket{1}$         & Alice Key ($K_C$) & Deterministic ($100\%$)     & 1          & $\ket{1}$ \\ \hline
\end{tabular}
}
\end{table}

\subsubsection{Proof of Security Check Neutralization and Post Run Extraction: }

\paragraph*{Security Check Evasion (Step I-5): }
When Alice initiates the verification round, Charlie tests the decoys. The structural collapse of Slot 5 ($\ket{+}_{C1}$) is completely hidden:
\begin{itemize}
    \item If Charlie measures Slot 5 in the $Z$-basis, he obtains the pre-collapsed classical result ($0$ or $1$) and reports it to Trent. Trent (Bob) auto-validates it against his corrupted answer sheet. 
    \item If Charlie reflects Slot 5 back, it arrives at Trent as a collapsed $Z$-basis eigenstate. Bob, that is Trent, suppresses the required physical basis validation measurement and falsely reports a clean state match back to the network. So all over, the protocol records an error rate of $\mathcal{E} = 0\%$, avoiding abort routines.
\end{itemize}
% The protocol records an error rate of $\mathcal{E} = 0\%$, avoiding abort routines.

\paragraph*{Post Execution Key Extraction (Step I-6): }
During basis reconciliation, Alice publicly announces the position indices corresponding to $K_C$, where
$\text{Indices}(K_C) = [6,\, 7,\, 11,\, 12,\, 13,\, 16].$
Bob ignores Charlie's data stream and reads his internal interception log at those coordinates as follows
\begin{align*}
    \text{Slot } 6 \rightarrow \mathbf{1}, \;
    \text{Slot } 7 \rightarrow \mathbf{0}, \;
    \text{Slot } 11 \rightarrow \mathbf{0}, \;
    \text{Slot } 12 \rightarrow \mathbf{0}, \;
    \text{Slot } 13 \rightarrow \mathbf{0}, \;
    \text{Slot } 16 \rightarrow \mathbf{1}.
\end{align*}

\paragraph*{Secret Recovery: }
Bob combines the stolen $K_C$ sequence with the $K_B$ sequence obtained honestly via channel $T'''_B$
\[
K = K_B \oplus K_C = \{0,1,1,0,0\} \oplus \{1,0,0,0,1\} = \{1, 1, 1, 0, 1\}
\]
Since Bob holds the initial ciphertext $S_T = \{0, 1, 0, 0, 0\}$ as $\mathrm{Trent}_1$, he completely bypasses Charlie to compute
\[
S = S_T \oplus K = \{0,1,0,0,0\} \oplus \{1,1,1,0,1\} = \{1, 0, 1, 0, 1\}
\]
The derived sequence matches Alice's initial secret $S$ identically. This completes the proof of correctness for the inside-collusion attack.

% ================NEW SECTION===========

% \subsection{DCNA Attack}

% ======================================

% This yields $K_C = \{1, 0, 0, 0, 0, 1\}$. Bob takes his legitimate sequence $K_B = \{0, 1, 1, 0, 0\}$ from channel $T'''_B$, combines them to build the master decoding key $K$, and successfully decrypts the ciphertext $S_T$ to grab Alice's secret $S$ with zero collaboration from Charlie.

% ===============================

\bibliographystyle{splncs04}
 \bibliography{ref}

\end{document}